%% file: main.tex
\documentclass[format=acmsmall]{acmart}

\usepackage[ruled]{algorithm2e} 

\SetAlFnt{\small}
\SetAlCapFnt{\small}
\SetAlCapNameFnt{\small}
\SetAlCapHSkip{0pt}
\IncMargin{-\parindent}

\usepackage{booktabs}   
\usepackage{subcaption} 
                        
\usepackage{amssymb,amsmath}
\usepackage{listings}
\usepackage{graphicx}
\usepackage{fancyvrb}
\usepackage{wrapfig}
\usepackage{longtable}
\usepackage{graphicx}

\usepackage{ stmaryrd }
\usepackage{proof}
\usepackage{amssymb}
\usepackage{wasysym}
\usepackage{ bbold }

\newcommand{\lolli}{\multimap}
\newcommand{\tensor}{\otimes}

\newcommand{\with}{\binampersand}
\newcommand{\one}{\mathbb{1}}

\newcommand{\btseq}[1]{\mathsf{Seq}\{#1\}}
\newcommand{\btsel}[1]{\mathsf{Sel}\{#1\}}
\newcommand{\btrepeat}[1]{\mathsf{Repeat}\{#1\}}
\newcommand{\btskip}{\btseq{}}
\newcommand{\btabort}{\btsel{}} 
\newcommand{\btcond}[2]{?{#1}{.}\;#2}
\newcommand{\btop}[2]{#1(#2)}

\newcommand{\fail}{\mathsf{FAIL}}

\newcommand{\eval}[3]{#1 \triangleright #2 \Downarrow #3}
\newcommand{\admits}{\Vdash}

\lstdefinelanguage{ceptre}{
	morekeywords = {type, pred, as, let, in, forall},
    otherkeywords = {->, :, -o, *, |, =, \\, !, $, ;},
    morekeywords = [1]{->, :, -o, *, |, =, \\, !, $, ;},
    morecomment = [l]{//}
}
\makeatletter
\newcommand*\idstyle{%
        \expandafter\id@style\the\lst@token\relax
}
\def\id@style#1#2\relax{%
        \ifcat#1\relax\else
                \ifnum`#1=\uccode`#1%
                        \bfseries
                \fi
        \fi
}
\makeatother

\begin{document}

\title{A Resourceful Reframing of Behavior Trees}


\author{Chris Martens}
\affiliation{
  \department{Computer Science Department}              
  \institution{North Carolina State University}            
  \city{Raleigh}
  \state{NC}
  \country{USA}                    
}
\email{martens@csc.ncsu.edu}          

\author{Eric Butler}
\affiliation{
  \department{Computer Science and Engineering}             
  \institution{University of Washington}           
  \city{Seattle}
  \state{WA}
  \country{USA}                   
}
\email{edbutler@cs.washington.edu}         

\author{Joseph C. Osborn}
\affiliation{
  \department{Computational Media}             
  \institution{University of California, Santa Cruz}           
  \city{Santa Cruz}
  \state{CA}
  \country{USA}                   
}
\email{jcosborn@soe.ucsc.edu}         

\begin{abstract}
Designers of autonomous agents, whether in physical or virtual environments, need to express nondeterminisim, failure, and parallelism in behaviors, as well as accounting for synchronous coordination between agents. Behavior Trees are a semi-formalism deployed widely for this purpose in the games industry, but with challenges to scalability, reasoning, and reuse of common sub-behaviors.

We present an alternative formulation of behavior trees through a language design perspective, giving a formal operational semantics, type system, and corresponding implementation. We express specifications for atomic behaviors as linear logic formulas describing how they transform the environment, and our type system uses linear sequent calculus to derive a compositional type assignment to behavior tree expressions. These types expose the conditions required for behaviors to succeed and allow abstraction over parameters to behaviors, enabling the development of behavior ``building blocks'' amenable to compositional reasoning and reuse.
\end{abstract}




\maketitle

\section{Introduction}
\label{sec:intro}
\input{intro}

\section{Related Work}
\label{sec:relwork}
\input{relwork}

\section{Background: Behavior Trees in Games}
\label{sec:background}
\input{background}

\section{Action Specifications in Linear Logic}
\label{sec:linear}
\input{linear}

\section{BTL: A Formal Semantics for Behavior Trees}
\label{sec:btl}
\input{btl}

\section{Compositional Reasoning}
\label{sec:types}
\input{types}

\section{Implementation}
\label{sec:impl}
\input{impl}


\section{Discussion}
\label{sec:discussion}
\input{discussion}

\section{Conclusion}
\label{sec:conclusion}
\input{conclusion}


\bibliography{main}

%

\end{document}

%% file: intro.tex


Specifying the desired behaviors of agents in environments is a major theme in artificial intelligence.  Analysts often need to define particular policies with explicit steps, but the agents must also acknowledge salient changes in a potentially hostile or stochastic environment.  This challenge arises in applications including robotics, simulation, and video game development. Games in particular bring challenges related to interaction with human decision-makers:  even for games notionally working against the objectives of the player, the activity of game design is centrally concerned with helping the player learn something or have an emotional experience, and in this sense can be thought of as a cooperative system between agents with different knowledge states, not unlike human-robot teams. The behaviors of non-player characters (NPCs) in games must be designed in support of this goal.

Game designers must be able to specify that a given agent should patrol a hallway until it gets hungry (or its battery runs low) and goes home for a snack (or to recharge); but if the agent sees a one-hundred dollar bill on the ground on the way to where it recuperates, it should force a detour.  In some designs, we would want an adversary (e.g., the player) to be able to trick the agent into running out of fuel by this mechanism; in other designs we would hope the agent ignores optional but attractive diversions and prioritizes severe need. We can easily imagine two distinct agents within the same game which are differentiated only by whether they can be misled in such a way.

Game character AI designers have somewhat contradictory goals that distinguish their project from, for example, game-playing AI whose objective is optimal play.  On the one hand they want \emph{believable} characters who react reasonably to player actions and to the behaviors of other non-player characters; but on the other hand they want to craft certain very specific \emph{experiences} that nudge the player into trying new actions or approaching problems from a specific direction or that prevent the agent from performing awkward-looking sequences of 3D animations.  Traditionally, game character AI was implemented with explicit state machines built by hand; more recently behavior trees, goal-oriented action planning, and utility-based systems have come into vogue.


\begin{figure}[ht]
\includegraphics[width=0.6\linewidth]{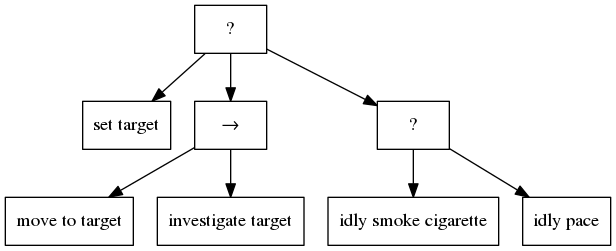}
\caption{%
A example behavior tree for a noise-investigation behavior. The tree is evaluated in preorder traversal. Leaf nodes specify actions in the world (such as moving to a target), which can succeed or fail. Interior nodes combine children into composite behaviors. The arrow ($\rightarrow$) is sequencing (run each child until first failure), and the question (?) is selection (run each child until first success).%
}
\label{fig:investigate-bt-viz}
\end{figure}

Behavior trees are a scripting system for agents in virtual worlds, allowing designers of virtual agents to visually construct behavioral flowcharts based on conditions on the world around them. They are widely employed in the video games industry~\cite{rabin13gameaipro} for describing the ``artificial intelligence'' behavior of non-player characters, such as enemy units in combat simulators and members of virtual populations in open world-style games. Behavior trees have also been used for robot control~\cite{marzinotto2014towards}. They are often described as merging the utility of decision trees and state machines, allowing repeated or cyclic behaviors that modify and check state (internal or shared) as they execute. Figure~\ref{fig:investigate-bt-viz} shows an example behavior tree for a hypothetical security guard character. The tree defines how to sequence and prioritize basic behaviors of listening for noises, investigating the source of the noise, or doing idle activities. During game simulation, behavior trees are typically re-executed with some frequency depending on the game, as often as once or more per time step. The example in Figure~\ref{fig:investigate-bt-viz}, for instance, needs to be executed twice to both acquire a target and investigate it.

Since behavior trees are often deployed in multi-agent simulations and with complex state-changing behavior, the ability for a designer to reason about the correctness of the tree quickly succumbs to its size and branching factor. Even for straightforward sequences of behaviors, the preconditions and postconditions are left unstated. For example, if an agent is told to \verb|move to door|, \verb|open door|, and \verb|go through door|, we might reasonably expect that in all circumstances where the door is accessible, the agent will be on the opposite side of it by the time its behavior finishes. However, this is not possible to conclude unless we reason both about the conditions and effects of the individual actions and how the effects of earlier actions are expected to connect to the conditions of later ones. Such a sequence of actions could fail, for instance, if the player were to intervene and close the door immediately after the agent opened it. Furthermore, the success of behaviors may depend on external conditions on the environment: an agent may expect another agent to have placed an important item that it needs, and the behavior is only correct on the condition that this dependency has been satisfied.

We describe an approach to reasoning compositionally about behavior trees
in such a way that they may be constructed in small units, typechecked
against an expected behavioral schema, and combined to form behaviors with
new, compositionally-defined types. The approach requires the author to
provide a linear logical specification of the atomic actions, i.e.\ the
leaves of the tree; types for complex expressions formed from these leaves
are derived from a linear logical interpretation of the behavior tree
operations (sequencing, selection, and conditions).
%
The present work can be seen as a way to regain some of the guarantees
given by reasoning about a behavior from start to finish without losing the
reactivity, which is the main benefit of using behavior trees over, for
example, static plan generation~\cite{ghallab2016automated}.

Since behavior trees are a relatively simple formalism repeatedly realized
in different incarnations, and since game developers are under somewhat
notorious pressure to ship products, there is no authoritative,
\emph{standard} version of behavior trees.  As alluded to above, a
recurring issue with behavior trees is resolving the apparent tension
between reacting to unexpected changes in the environment on the one hand
and to performing authored behaviors over a longer duration on the other
hand.  The ad hoc extensions applied to behavior trees in the wild are
often intended to resolve this tension.  The approaches described in this
paper could give a theoretical foundation for addressing these ``hacks''
employed in practice---and, potentially, for more principled and
better-behaved adaptations of behavior trees towards the problem of
designing complex agent and character behaviors.

Our contributions are a formal specification and operational semantics for
our formulation of behavior trees, a type system and synthesis algorithm
backed by an interpretation in linear logic, and an implementation of these
systems in Standard ML.\@ These results represent the first step of toward
building a toolkit for robust authoring of virtual agent behaviors,
combining support for correct human authorship and certified goal-driven
synthesis of behaviors.

The rest of the paper is organized as follows: Section~\ref{sec:relwork}
discusses related work; Section~\ref{sec:background} describes further how
behavior trees are used in the games industry, Section~\ref{sec:linear}
explains linear logical specifications and how they may be used to describe
a possibility space for virtual worlds; Section~\ref{sec:btl} describes the
syntax and operational semantics of our behavior tree language;
Section~\ref{sec:types} describes the type system and its guarantees;
Section~\ref{sec:impl} describes our implementation;
Section~\ref{sec:discussion} discusses our current scope and future work;
and ~\ref{sec:conclusion} summarizes our contributions.


%% file: relwork.tex
For the most part, efforts to provide robust formalisms to designers of
virtual agents have been disjoint from formal and language-based
approaches.  We identify related work in two key areas: previous attempts
to characterize virtual agent behaviors from a formal methods standpoint,
and related models of computation that have been characterized with linear
logic.

\subsection{Formal accounts of behavior trees }


Marzinotto et al. provide an account~\cite{marzinotto2014towards} of
behavior trees in the context of robot control, citing a dearth of
mathematical rigor prior to their contribution. Their work contributes
the first mathematical definition of behavior trees and accounts for their
expressive capabilities. 

More recently, there has been some very recent work in applying synthesis and
verification to AI behavior trees~\cite{colledanchise2017synthesis}. The formal
basis for said work is model checking in linear temporal logic (LTL). Our work,
by contrast, seeks a type-theoretic solution that supports modular reuse of
behaviors. 

\subsection{Linear logical accounts of agents and processes}

Linear Session Types~\cite{caires2016linear} are an important touchstone for
this work as another characterization of a pre-existing system,
$\pi$-calculus, under a semantics derived from linear sequent calculus. Our
work does not identify a direct logical correspondence between logical and
operational notions in the same way, but similarly provides a basis for
static reasoning about complex behaviors.

The CLF~\cite{watkins2003concurrent} logical framework and corresponding
implementation Celf~\cite{schack2008celf} form a basis for interpreting
linear logic formulas as programs under a proof-construction-as-execution
paradigm (logic programming). While operationally, this approach diverges
from the semantics of behavior trees, the representation formalism informs
out approach.

Finally, linear logic has been used to account for planning in a number of
interesting ways: {\em deductive planning\/}~\cite{cresswell1999deductive}
runs with the observation that, in addition to Masseron et al.'s observation that 
linear proof search can model planning~\cite{masseron1993generatingI},
linear proofs {\em generalize\/} plans: they can characterize recursive and
contingent (branching) plans, recovering some of the same expressiveness as
behavior trees. Dixon et al.~\cite{dixon2006planning} apply deductive
planning to an agent-based domain for dialogue-based environments. This
work encourages us to consider integrating the generative abilities of
planners with the reactivity of behavior trees in future work.

%% file: background.tex
Behavior trees are widely used to define the behavior of non-player
characters in digital game genres ranging from strategy and simulation to
first-person shooters. The major game-making tools (Unreal Engine, Unity
3D, CryEngine, Amazon Lumberyard, and others) all either provide natively
or have third-party implementations of the technique.
The canonical examples of behavior trees' use in games come from the
\textit{Halo} series of first-person shooter games~\cite{isla2005halo}.
Notable in their formulation is that most of the tree is shared across the
different types of enemy agents that appear in the game, which reflects the
difficulty of authoring good and reasonable behavior policies in general.
Behavior trees give authors a way to reuse some behaviors and override
others from agent to agent.

Behavior trees are usually characterized as a \emph{reactive} AI formalism,
in this context meaning that agents are defined in terms of their reactions
to a changing environment, rather than by a top-down plan that tries to
achieve a goal by considering contingencies in advance.  Certainly, even
finite state machines can be made reactive by adding appropriate
transitions, but scaling them to myriad potential game events quickly
overwhelms authors. Behavior trees reduce that burden by asking a behavior
author to structure the reactive behaviors in a tree, implicitly defining
which behaviors supersede or interrupt which other behaviors by their
position in a preorder traversal of that tree.

A behavior tree is a data structure describing how an agent decides on its
next actions, and at the leaves some primitives for executing those
actions.  Behavior trees are repeatedly \emph{evaluated} and on each
evaluation they process their nodes in sequence.  When a node is processed,
it evaluates with some status: \texttt{RUNNING}, \texttt{SUCCEEDED}, or
\texttt{FAILED}.  Different sorts of nodes in the tree are specified in
terms of the circumstances under which they evaluate to each return value.

A key question in behavior tree semantics is whether a tree which ends an
evaluation with the \texttt{RUNNING} status should, on the next evaluation,
continue from where it left off; the alternative is for it to begin its
next evaluation from the root again.  The latter approach is more
\emph{reactive} to changes in the environment or interruptions to
behaviors, but in the former it is easier to specify and conceptualize
behaviors which take some time and should not be interrupted.  It is also
easier to avoid behavior \emph{oscillations} in the former evaluation
strategy. For example, with the investigation example from
Figure~\ref{fig:investigate-bt-viz}: with the latter approach, the agent
can be interrupted by a new noise when moving to a target, while with the
former approach, the agent will fully investigate a target without
distraction. Game designers have explored both semantics and even hybrids
between these approaches; we leave our discussion of this issue until
Sec.~\ref{sec:btl}.

Leaf nodes of the tree can be domain-specific \emph{conditions} (which
succeed if the condition is currently satisfied or fail otherwise) or
domain-specific \emph{actions} (for example, setting the value of a
variable or triggering some external action). These are the only operations
which can interact with the environment. The actions in
Figure~\ref{fig:investigate-bt-viz} include setting a variable representing
the agent's current target or physically navigating the agent towards said
target. Failure may come from, for example, there being no navigable path
to the target. In video games, these are often implemented using arbitrary
program code working outside of the behavior tree formalism.

Non-leaf nodes come in three key variants (others are usually possible to
define as syntactic sugar).  First, \emph{sequences} evaluate each of their
child nodes from left to right, and are \texttt{RUNNING} if any child node
is \texttt{RUNNING}, \texttt{FAILED} if any child is \texttt{FAILED}, or
\texttt{SUCCEEDED} otherwise.  Second, \emph{selectors} also evaluate their
child nodes left to right, but are \texttt{RUNNING} if any child is
\texttt{RUNNING}, \texttt{SUCCEEDED} if any child has \texttt{SUCCEEDED},
and \texttt{FAILED} if all the child nodes are \texttt{FAILED}.  Third, the
\emph{parallel} node evaluates each of its children independently of each
other, and has \texttt{SUCCEEDED} if more than a certain number of its
children succeeds, \texttt{FAILED} if more than a certain number fail, and
\texttt{RUNNING} otherwise.  It is also implicit in the definition of
behavior trees that there is some external environment where state can be
stored and persisted from evaluation to evaluation.  


In practice, there are many other types of nodes that can alter the
semantics of the tree in arbitrary ways, often violating the assumption of
a preorder traversal:  repeaters which evaluate their children over and
over until they evaluate with some status, stateful versions of sequence
and selector with \emph{memory} that remember when they get stuck in
\texttt{RUNNING} and only evaluate from that stuck node forwards in their
next evaluation, and so on.  We ignore such extensions in this work to
simplify the presentation.  Most of the extensions of behavior trees are
meant to facilitate long-running actions, to limit the reactivity of
behavior trees (e.g., to allow interruptions only at designer-defined
times), and to ease the sharing of behavior tree or character state across
situations, characters, and actions.  Actions, conditions, and decorators
often themselves involve arbitrary code in practice, so in our presentation
of the formal semantics we require a linear logic formulation of the leaf
nodes.


%% file: linear.tex
As a first step towards a type system for general behaviors, we concretize
{\em action specifications} for describing the behavior of an atomic
action, such as ``idly smoke cigarette'' in
Figure~\ref{fig:investigate-bt-viz}. 
Although in reality, this behavior may simply take the form of an
observable effect (e.g., some animation), semantically, there are certain
things we expect for it to make sense: for instance, that the agent has a
supply of cigarettes (and perhaps that this action spends one). Other
actions, like passing through a door, have more important requirements and
effects, such as requiring being near the door and resulting in the door
being open: these are aspects of the environment that may be created, or
depended on, by other agents (or the same agent at another time).

There is a long line of successful work on describing actions in a
protocols and virtual worlds using any of a class of related formalisms:
multiset rewriting, Petri nets, vector addition systems, and linear logic.
These systems have in common an approach to specification using {\em rules}
(or {\em transitions} in some systems) that describe dependencies and
effects, such that the cummulative effects of applying those rules may be
reasoned about formally.\footnote{Planning domain description languages
 also share this approach, but most standards such as PDDL~\cite{mcdermott1998pddl},
 do not have as clean of a compositional interpretation due to their allowance for
the ``deletion'' of facts that do not appear as preconditions.}

The following example uses a linear logic-based notation adapted from
Ceptre~\cite{martens2015ceptre} to describe action specifications for an
\texttt{Investigation} world that could assign meaning to the actions used
in Figure~\ref{fig:investigate-bt-viz}:

\begin{lstlisting}[frame = single]
set_target       : no_target -o has_target.
move_to_target  : has_target -o has_target * at_target.
investigate      : has_target * at_target * heard_noise -o no_target.
smoke            : has_cigarette -o 1.
pace             : 1 -o 1.
\end{lstlisting}

The ``lolli'' syntax \verb|A -o B| describes the ability to transition from a world
in which \verb|A| obtains to one in which \verb|A| no longer obtains and
has been replaced with \verb|B|. The atomic propositions include facts like
\verb|at_door| and \verb|door_open|, which represent pieces of world state,
the ``tensor'' \verb|p * q| syntax conjoins them, and \verb|1| is the unit of
tensor. 
World configurations can be represented as multisets (or linear contexts)
$\Delta$ specifying which facts hold, such as 
\lstinline|{no_target, heard_noise, has_cigarette, has_cigarette}|.

\begin{figure} 

\centering
\includegraphics[width=0.6\textwidth]{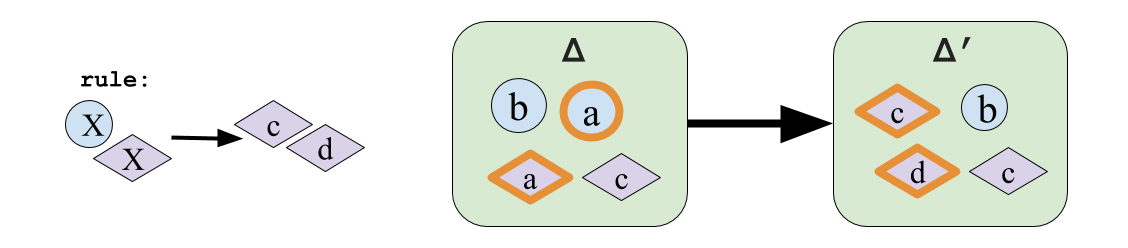}

\caption{One step of multiset rewriting execution, visualized. 
Each color/shape (purple diamond, blue circle) represents a distinct predicate; the contents of those shapes are terms (a,
b, c) or term variables (X). This diagram represents a transition of the
state $\Delta = \{diamond(a), circle(a), circle(b), diamond(c)\}$ along the rule
$circle(X) \tensor diamond(X) \lolli diamond(c) \tensor diamond(d)$ to the new
state $\Delta' = \{diamond(c), diamond(d), circle(b), diamond(c)\}$. The thick orange borders on some atoms highlight which ones are replaced and added by the rule.  }

\label{fig:msrw}
\end{figure}

In general, predicates can take arguments (e.g., \lstinline{at(castle)}) and
rules can universally quantify over variables that stand in for term
arguments, in which case states are always ground (contain no variables)
and the application of rules identifies appropriate substitutions for
variables for which the rule applies. Figure~\ref{fig:msrw} visualizes a
step of execution for an example.

Multiset rewriting has been used commonly to model nondeterminism and
concurrency: rulesets can be nondeterministic whenever multiple rules may
apply to a given state, and concurrency arises from the partial-order
causal relationships between rules firing. If two rules operate on disjoint
parts of the state, for instance, they can be considered to fire
simultaneously, whereas rules that depend on the effects of previous rules
firing must obey sequential ordering. See Figure~\ref{fig:causal} for a
diagram of the causal relationships between actions for a particular
program trace in which the agent sets a target, moves to the target,
investigates a noise, and smokes a cigarette.

For the work described in this paper, however, we are less interested in
the multiset rewriting interpretation of the rules.  The specification
under the multiset rewriting interpretation alone does not give us as
authors any control over strategies for action selection or goal-driven
search. Instead, it can be thought of as a description of the {\em space of
possible actions} and a way of calculating their cumulative effects.
Behavior trees, then, can be understood as directives for how to explore
this space.

\begin{figure} 

\centering
\includegraphics[width=0.9\textwidth]{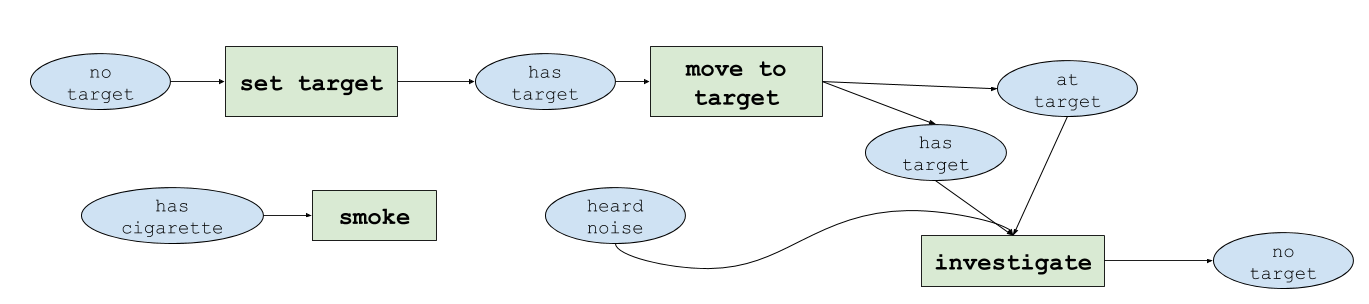}

\caption{A causal diagram for a possible trace of actions under the
  multiset rewriting interpretation of the \texttt{Investigate}
specification.}
\label{fig:causal}
\end{figure}


Formally, we define action specifications under the following grammar:
\begin{eqnarray*}
  arg   &::=& t \mid x \\
  args  &::=& \cdot \mid arg, args \\
  S     &::=& p(args) \mid \one \mid S \tensor S \\ 
  opdecl    &::=& name : xs{.}\; S \lolli S\\
  \Sigma &::=& \cdot \mid \Sigma, opdecl
\end{eqnarray*}

$\Sigma$ is a collection of specifications for operators $op$. $\Sigma$
may also specify a collection of valid domain types over which the
arguments of operators may range; for example, the operator $move(Dir,N)$ may
range over directions and natural numbers, perhaps meaning to move in that
direction a certain number of units. 
The world state $\Delta$ is represented as a linear logic context, i.e. a
multiset of atomic propositions $p(args)$ representing available resources. 

In the next section, we assume an arbitrary signature $\Sigma$ for each
action that computes a function on states, which does not depend on the
linear logical interpretation. However, we revisit this idea in
Section~\ref{sec:types} to assign types to behavior tree expressions.

%% file: btl.tex
In this section we describe BTL, a formal calculus for describing
synchronous agent behaviors with sequencing, branching, conditions, and
loops.

The goals of this system are similar in many ways to the BTs used in
practice: we aim to provide simple authoring affordances for scripting
reactions to different circumstances in an implicit environment that is
changing around the agent, and which the agent can change. 
We also adopt some goals that are not currently met by industry practice:
\begin{itemize}
  \item Compositional reasoning. In order to be able to reason about BT
    nodes in terms of the behaviors of their subtrees, we need to know that
    subtree behaviors won't be interrupted in unknowable states.
  \item Debugging support---specifically, the ability for authors to state
    what they expect a behavior to accomplish and have this expectation
    checked algorithmically. The algorithm should be able to identify where
    in the tree a stated expectation is violated.
  \item Support for the expression of coordinated multi-agent behaviors.
    This requirement means that we question the notion of an action
    dependency necessarily meaning {\em failure} and instead (or
    additionally) require a {\em blocking} semantics (for instance, an
    agent may wait until another agent joins them in the same location to
    hand off a needed item).
  \item Support for the eventual integration of {\em behavior synthesis},
    or algorithms that accept a propositional goal for the world state and
    generate BT subtrees corresponding to conditional plans that achieve
    the goal.
\end{itemize}

These nonstandard goals entail 
some tradeoffs of
expressiveness. While it would be ideal to retain, for example, the
``reactive'' nature of BTs that allow them to break sequential actions to
tend to urgent interruptions, we do not adopt this form of reactivity by
default because it would preclude the ability to reason about sequential
behaviors compositionally. In Section~\ref{sec:discussion} we revisit these
expressiveness tradeoffs and consider ways to re-incorporate additional
features.

%
%
%
%

\subsection{Expressions}

The expressions of BTL are:

\vspace{-1em}
\newcommand{\actpar}[2]{{#1}{\parallel}{#2}}
\begin{eqnarray*}
  \alpha &::=&  \btop{op}{args} 
                \mid \btcond{p}{\alpha} 
                \mid \btseq{\alpha; \alpha}
                \mid \btsel{\alpha + \alpha}
                \mid \btskip 
                \mid \btabort
                \mid \btrepeat{\alpha}
\end{eqnarray*}

Intuitively, $\btop{op}{args}$ is an atomic
action, invoking a pre-defined operator on a set of ground arguments (such
as \lstinline{move(left)}); 
$\btseq{\alpha;\alpha}$ is a sequence node;
$\btseq{\alpha + \alpha}$ is a selector node; 
$\btskip$ is the unit of sequencers (does nothing);
$\btabort$ is the unit of selectors (always fails);
$?p.\; \alpha$ checks the condition $p$ and executes $\alpha$ if it holds,
failing otherwise;
and $\btrepeat{\alpha}$ is a repeater node, running $\alpha$ until failure. 

\subsection{Operational Semantics}

We define an operational semantics for behavior trees in terms of what they
may do to an abstract world state, using a big-step evaluation judgment
$\eval{\alpha}{\Delta}{\delta}$, where $\Delta$ is a world state and
$\delta$ is the result of evaluating a BTL expression, either a new world
state (on successful execution) or $\fail$.

The evaluation judgment requires a few preliminaries to define. First, we
implicitly index the judgment by a signature $\Sigma$, which provides a
specification for a {\em transition function} $t : \tau \to \Delta \to
\delta$ for each operator (atomic action) available to an agent, which
takes arguments of type $\tau$, computes a transformation on a world state
if the action can be performed, and returns $\fail$ otherwise.  Concretely,
our linear logical action specifications can play this role.
Second, we assume
a notion of a condition ``holding for'' a world state, expressed by the
judgment $\Delta \admits p$. Again, while evaluation can be defined holding
this judgment abstract, in we can fulfill this definition by expressing
conditions in terms of a (positive) subset of linear logic formulas and
interpreting $\admits$ as affine provability.


Evaluating an operation consists of looking up its transtition function in
$\Sigma$ and applying that function to the current state; evaluating a
condition requires that the current state satisfies the condition, and
otherwise fails:

\[
  \infer
  { \eval{\btop{op}{args}}{\Delta}{\delta} }
  { \Sigma(op) = t
    &
    t(args, \Delta) = \delta
  }
  \qquad
  \infer
  { \eval{\btcond{S}{\alpha}}{\Delta}{\delta} }
  { \Delta \admits S &
    \eval{\alpha}{\Delta}{\delta}
  }
  \qquad
  \infer
  { \eval{\btcond{S}{\alpha}}{\Delta}{\fail} }
  { \Delta \not{\admits} S }
\]

A sequence evaluates by chaining the states computed by successful subtrees
through successive subtrees in the sequence, and fails if any subtree
fails:

\[
  \infer
  { \eval{\btskip}{\Delta}{\Delta} }
  {}
  \qquad
  \infer
  { \eval{\btseq{\alpha; \alpha'}}{\Delta}{\delta} }
  { \eval{\alpha}{\Delta}{\Delta'}
    &
    \eval{\btseq{\alpha'}}{\Delta'}{\delta}
  }
  \qquad
  \infer
  { \eval{\btseq{\alpha; \alpha'}}{\Delta}{\fail} }
  { \eval{\alpha}{\Delta}{\fail}
  }
\]

A selector succeeds with the first successful subtree and fails if no
options are possible:

\[
  \infer
  { \eval{\btabort}{\Delta}{\fail} }
  {}
  \qquad
  \infer
  { \eval{\btsel{\alpha + \alpha'}}{\Delta}{\delta} }
  { \eval{\alpha}{\Delta}{\fail}
    &
    \eval{\btsel{\alpha'}}{\Delta}{\delta}
  }
  \qquad
  \infer
  { \eval{\btsel{\alpha + \alpha'}}{\Delta}{\Delta'} }
  { \eval{\alpha}{\Delta}{\Delta'} }
\]

Repeaters continue evaluating the underlying expression until failure:

\[
  \infer
  { \eval{\btrepeat{\alpha}}{\Delta}{\delta} }
  { \eval{\alpha}{\Delta}{\Delta'}
    &
    \eval{\btrepeat{\alpha}}{\Delta'}{\delta}
  }
  \qquad
  \infer
  { \eval{\btrepeat{\alpha}}{\Delta}{\Delta} }
  { \eval{\alpha}{\Delta}{\fail}
  }
\]

This definition of BTL adopts similar conventions and semantics to process
algebras, such as the adoption of two key operators, sequential
(conjunctive) and choice (disjunctive) composition, which have certain
algebraic properties. In the case of BTL, evaluation respects the following
structural congruence:

\renewcommand{\cong}{\equiv}

\begin{eqnarray*}
  \btseq{\btskip;\alpha}  &\equiv \btseq{\alpha} \equiv& \btseq{\alpha; \btskip}\\
  \btseq{\alpha; \btseq{\beta; \gamma}} &\equiv&
      \btseq{\btseq{\alpha;\beta}; \gamma} \\
  \btsel{\btabort + \alpha} &\equiv \btsel{\alpha} \equiv& \btsel{\alpha + \btabort}\\
  \btsel{\alpha + \btsel{\beta + \gamma}} &\equiv&
      \btsel{\btsel{\alpha + \beta} + \gamma} \\
  \btseq{\alpha; \btsel{\beta + \gamma}} &\equiv&
      \btsel{\btseq{\alpha; \beta} + \btseq{\alpha; \gamma}}\\
  \btseq{\btsel{\alpha + \beta}; \gamma} &\equiv&
      \btsel{\btseq{\alpha; \gamma} + \btseq{\beta; \gamma}}
\end{eqnarray*}

In other words, sequences form a monoid with the unit $\btskip$; selectors
form a monoid with the unit $\btabort$; and sequencing distributes over
selection. We state that this equivalence is respected by evaluation but
omit the proof for brevity:

\begin{conjecture}
BTL operational semantics respects congruence:
If $\eval{\alpha}{\Delta}{\delta}$ and $\alpha \cong \beta$ then
 $\eval{\beta}{\Delta}{\delta}$.
\end{conjecture}

While the system bears resemblance to models of concurrency such as
CSP~\cite{hoare1978communicating} and (CCS)~\cite{milner1980calculus}, it
differs in that interactions between BTL expressions and their environment
happen implicitly through manipulation of a shared world state, not through
channel-based communication (as in CSP) or explicit labels for inputs and
outputs (as in CCS). The lack of such machinery is what makes behavior
trees so attractive to authors; it reduces the burden of needing to specify
how information is transmitted from one agent to another. However, it also
makes the dependencies between agents tacit and therefore difficult to
debug when things go wrong, which is what this paper aims to address.

Kleene algebra, particularly Kozen's variant with tests
(KAT)~\cite{kozen1997kleene}, offers
another touchstone for semantic insights; however, BTL does not quite
satisfy the Kleene conditions: (1) order matters in
selector semantics due to fallthrough, so selectors are not commutative;
(2) the annihilation law does not hold; $\btseq{\alpha; \btabort}$ is not
equivalent to $\btabort$ due to the state changes that $\alpha$ may incur.

\subsection{Example}
\label{sec:eval-example}

Below is BTL implementation of the behavior tree described in Figure~\ref{fig:investigate-bt-viz}. This and all future examples use an n-ary form of Seq and Sel defined in the obvious way.
\input{figures/intro-ex.tex}
To illustrate how an BTL expression evaluates, we consider an evaluation of
this tree in an environment where the agent already has a reachable target
and has not heard a noise, i.e. the situation 
$\{\mathtt{has\_target}\}$. 
Starting evaluation at the root, the outer selector expression evaluates
each child in succession until one succeeds. The first child will fail
because the \texttt{heard\_noise} condition does not hold. The second
child, a sequence, will evaluate each of its children in succession. The
first action, predicated on having a target, evaluates by modifying the
world state such that the agent is in the same location as the target. Upon
the movement action succeeding, the \texttt{investigate\_target()} action
will be evaluated; however, this node fails in the absence of having heard
a noise, and that failure propagates to the root of the tree.

If instead we started in an environment
$\{\mathtt{has\_target},\mathtt{heard\_noise}\}$, then at the same point in
the tree, the \texttt{investigate\_target} action will succeed and
change the world state by replacing \texttt{has\_target} with
\texttt{no\_target} (in practice, this might have a more interesting effect
like updating variables representing the agent's knowledge of its target).
Because both children of the sequence evaluate to success, the sequence
evaluates to success. Thus, the root selector will itself evaluate to
success without evaluating the third branch, completing the evaluation of
the entire tree, and resulting in the state $\{\mathtt{no\_target}\}$.

%% file: figures/intro-ex.tex
\begin{lstlisting}[frame = single]
Sel{?heard_noise.set_target() +
    Seq{move_to_target(); investigate_target()} +
    Sel{idle_smoke() + idle_pace()}}
\end{lstlisting}

%% file: types.tex
\newcommand{\mseq}{\mathsf{seq}}

Compositional reasoning for behavior trees means that understanding the
effects of a whole BT can be done by understanding the effects of its
subtrees. The type system we describe gives a precise account of the
conditions under which a BT has {\em successful} execution and the
consequences of that execution. Accounting for the range of behaviors
possible under failure is outside the scope of this paper (see Section
~\ref{sec:discussion}).  However, these types are richer than sets of
precondtions and postcondtions: they account for the ``reactive'' nature of
BTs by requiring dependencies to be filled not prior to execution but just
at the node of the tree where they are needed; types also describe
resources that are {\em released} periodically if they are not needed for
later use.

This ``open'' structure of behavior types makes the account of agents'
behavior amenable to analysis in the presence of multiple agent executing
in parallel: BTs may both incur and use changes in the environment.

\subsection{A linear type system, take 1}
\label{sec:interfaces}

\newcommand{\D}{\Delta}
\newcommand{\G}{\Gamma}
\renewcommand{\S}{\Sigma}
\newcommand{\lseq}[3][]{{#2} \vdash_{#1} {#3}}

\begin{figure} 
\fbox{%
\parbox{0.98\linewidth}{%
\[
\infer[init]
{
 \lseq{\G;p}{p}
}{}
\qquad
\infer[\one R]
{
 \lseq{\G;\cdot}{{\bf 1}}
}{}\qquad
\infer[{\bf 1}L]
{
 \lseq{\G;\D, {\bf 1}}{C}
}
{
 \lseq{\G;\D}{M:C}
}
\qquad
\infer[\top R]
{
  \lseq{\G;\D}{\top}
}
{}
\qquad
\mathrm{(no\ }\top L\mathrm{)} 
\]

\[
\infer[{\tensor}R]
{
 \lseq{\G;\D_1, \D_2}{A \tensor B}
}
{
 \lseq{\G;\D_1}{A}
\quad
 \lseq{\G;\D_2}{B}
}\qquad
\infer[{\tensor}L]
{
 \lseq{\G;\D, A \otimes B}{C}
}
{
 \lseq{\G;\D, A, B}{C}
}
\]

\[
\infer[{\lolli}R]
{
 \lseq{\G;\D}{A \lolli B}
}
{
 \lseq{\G;\D,A}{B}
}
\qquad
\infer[{\lolli}L]
{
 \lseq{\G;\D_1, \D_2,A \lolli B}{C}
}
{
 \lseq{\G;\D_1}{A}
\quad
 \lseq{\G;\D_2,B}{C}
}
\]

\[
\infer[{\with}R]
{
 \lseq{\G;\D}{A \with B}
}
{
 \lseq{\G;\D}{A}
\quad
 \lseq{\G;\D}{B}
}
\qquad
\infer[{\with}L_i]
{
 \lseq{\G;\D,A_1 \with A_2}{C}
}
{
 \lseq{\G;\D,A_i}{C}
}
\]

\[
\infer[\forall{R}]
{
 \lseq{\G;\D}{\forall{x{:}\tau}.\; A}
}
{
 \lseq{\G, x{:}\tau;\D}{A}
}
\qquad
\infer[\forall{L}]
{
 \lseq{\G;\D, \forall{x{:}\tau}.\;A}{C}
}
{
 \G \vdash t:\tau
 &
 \lseq{\G;\D,N[t]{:}A}{C}
}
\]

}}
\caption{A fragment of intuitionistic linear sequent calculus.}
\label{fig:dill}
\end{figure}

Our guiding principle for assigning types to BTL expressions adopts a
``formulas-as-processes'' point of view to imagine the proof-theoretic
semantics of what the formula admits provable under arbitrary environments.
Consider linear logic formulas 
$A ::= p \mid \one \mid \top \mid A \tensor A \mid A \with A \mid A \lolli
A$ and an intuitionistic sequenct calculus defining their provability
(following ~\cite{chang03judgmental}) shown in Figure~\ref{fig:dill}.

The following intuition guides the correspondence we seek:

\begin{itemize}
  \item Firing a leaf action $\btop{op}{args}$ of 
      type $S \lolli S'$ in an environment $\Delta$
      corresponds to the $\lolli$-left rule in linear sequent calculus: to
      succeed, it requires that the current environment match the
      antecedent of the action and then changes the environment to replace
      it with the consequent. Correspondingly,
      evaluating
      $op(args)$ in an environment $\Delta, \Delta'$ where $\Delta' \vdash
      S$ evaluates to $\Delta, S'$ in the operational semantics.
  \item The unit selector $\btabort$ always fails, having run out of
    options; this corresponds to the $\top$ unit of $\with$ in linear
    logic, which has no left rule, so everything is beneath it in the
    preorder.
  \item The unit sequence $\btskip$ does nothing, corresponding to the left
    rule of the unit $\one$ of $\tensor$.
    Correspondingly, the operational semantics of $\btskip{}$ take the
    environment $\Delta$ to itself.
  \item Selectors $\btsel{\alpha_1 + \alpha_2}$ {\em nearly\/} correspond to
    making a choice, as in Linear Logic's $\with$ operator.
    There is a difference in that $\with$ is symmetric; $A\with B$ and $B
    \with A$ are interprovable, whereas order matters in BTL selectors.
    However, certain reasoning principles apply: if either
    $\eval{\alpha_1}{\Delta}{\Delta_1}$ or
    $\eval{\alpha_2}{\Delta}{\Delta_2}$, then one of $\Delta_1$ or
    $\Delta_2$ will be the result of evaluating $\btsel{\alpha_1 + \alpha_2}$
    against $\Delta$.
\end{itemize}

For reasons described above, however, accounting for sequences will be more
difficult. It might be tempting to think that $\tensor$ is an appropriate
interpretation, despite the relative lack of ordering constraints, for
reasons similar to how $\with$ can approximate selectors. A conjectured
rule:

\[
  \infer[BAD\_RULE]{\btseq{\alpha_1;\alpha_2} : A_1 \tensor A_2}
    {\alpha_1 : A_1 & \alpha_2 : A_2}
\]

At this point we need to formulate the metatheorem we have so far been
implicitly expecting to hold:

\begin{conjecture}
\label{conj:typesound}
If $\alpha : A$ and $\Delta, A \vdash S$, then
$\eval{\alpha}{\Delta}{\Delta'}$ and $\Delta' \vdash S$.
\end{conjecture}

(Recall that $S$ stands for a formula with no $\with$s or $\lolli$s,
representing a {\em successful} state in the course of a BTL expression's
execution.)
The proposed rule violates this conjecture; we show a counterexample next.

\subsection{The trouble with sequences: an example}

The following action specification describes a \texttt{Doors} world
in which agents may pass through open doors, open unlocked doors, and
unlock locked doors if they have keys:

\begin{lstlisting}[frame = single]
walk_to_door : at_elsewhere -o at_door.
pass_through : door_open * at_door -o door_open * through_door.
open_door    : door_unlocked * at_door -o door_open * at_door.
smash_door   : door_locked * at_door -o door_open * at_door.
close_door   : door_open * through_door -o door_unlocked * through_door.
\end{lstlisting}

\renewcommand{\a}{\mathtt{at\_elsewhere}}
\renewcommand{\b}{\mathtt{at\_door}}
\renewcommand{\c}{\mathtt{door\_unlocked}}
\renewcommand{\d}{\mathtt{door\_open}}

For a counterexample to Conjecture~\ref{conj:typesound}, let
$\alpha = \btseq{\mathtt{open\_door}; \mathtt{walk\_to\_door}}$ and let
$\Delta = \{\mathtt{at\_elsewhere}, \mathtt{door\_unlocked}\}$. According
to $BAD\_RULE$, $\alpha : A = (\b \tensor \c \lolli \d)\tensor(\a \lolli \b)$.
By straightforward rule applications, $\Delta, A \vdash \c$, but
it is not the case that $\eval{\btseq{\mathtt{open\_door};
\mathtt{walk\_to\_door}}}{\Delta}{\c}$.

%

In addition to the clear unsoundness of describing a sequential behavior
with a commutative connective, there are also concerns regarding the
granularity of concurrent execution.
Consider a simple sequential behavior for opening and going through a door:
\begin{lstlisting}[frame = single]
Seq{walk_to_door; open_door; pass_through; close_door}
\end{lstlisting}

A type we could reasonably expect to ascribe to this behavior is:

\[ \mathtt{at\_elsewhere} \tensor \mathtt{door\_unlocked} \lolli 
  \mathtt{through\_door} \tensor \mathtt{door\_unlocked} \]

This formula corresponds to the assumption that
if our starting environment has \texttt{at\_elsewhere} and
\texttt{door\_unlocked}, each element in this sequence of actions will
consume the output of the previous action as an input, resulting in
\texttt{through\_door}.  
Each successive action depends on the effects of previous actions: opening
the door assumes that the previous \verb|walk| action brought us to the
door; passing through assumes we successfully opened the door; and closing
the door assumes we passed through and the door is still open.

However, in a general, maximally concurrent environment, we would not be allowed to
make these assumptions: suppose, for example, another agent interferes and
closes the door just after we open it. This relaxed assumption instead
observes that we might forfeit {\em all} of the effects of previous
actions, resulting in the following type:

\begin{eqnarray*} 
  \mathtt{at\_elsewhere} \lolli \mathtt{at\_door} \tensor  
  (\mathtt{at\_door} \tensor \mathtt{door\_unlocked} \lolli & \\
  \mathtt{at\_door} \tensor \mathtt{door\_open} \tensor & \\
    (\mathtt{at\_door} \tensor \mathtt{door\_open}
    \lolli \mathtt{through\_door} \tensor \mathtt{door\_open} \tensor & \\
    (\mathtt{door\_open} \tensor \mathtt{through\_door} \lolli
    \mathtt{through\_door} \tensor \mathtt{door\_unlocked})))
\end{eqnarray*}

This formula characterizes the behavior that, at each step, a sequence
releases some resources into the world along with a ``continuation'' that
could, in some cases, potentially reabsorb those resources, or require new
ones along the way.

These two ascriptions correspond to different assumptions about
how behaviors interact with other behaviors manipulating the environment.
The former assumes an un-interruptable, ``critical section'' behavior to
sequences and gives a stronger guarantee, allowing us to treat the sequence
as a black-box behavior without worrying about internal failure. On the
other hand, the latter permits interruption and ``race condition''-like
scenarios that are common in games and interactive simulations in practice,
but offers less strict guarantees that reflect the complexity of reasoning
about fine-grained interaction.

Our type system makes the latter assumption that processes may be
interrupted, but we discuss the potential to accommodate both in
Section~\ref{sec:discussion}.


\subsection{Linear Behavior Interfaces}

We constrain linear logical formulas to the following grammar of
interfaces, expressed inductively as nested \emph{stagings} of inputs and
outputs (and choice between multiple possible interfaces):

\begin{eqnarray*}
N &::=& S \mid S \lolli N \mid S \tensor N \mid N \with N \mid \top
\end{eqnarray*}

This grammar mainly serves to prevent $\lolli$ from appearing to the left
of another $\lolli$ while representing staged inputs and outputs as
described above.
 
We assign types as linear logic formulas $N$ to BTL expressions
$\alpha$ with the judgment $\alpha :_{\Sigma} N$.  where $\alpha$ is an
expression, $N$ is an interface type, and $\Sigma$ is a specification
giving types $S \lolli S'$ to the actions used at the leaves of the trees.

The typing rules are as follows, with $\Sigma$ left implicit as an index to
the judgment except when it is needed. Atomic operations, conditions, 
the units, and selectors, are straightforward, and conditions must assume,
but then reproduce, the condition they depend on. Sequences are assigned a
type based on a computation $\mseq$ of the types of their components:

\[
  \infer
  {\btskip : \one}
  {}
\qquad
\infer
{\btabort : \top}
{}
\qquad
\infer
{op(args) :_{\Sigma} [args/xs] (S \lolli S')}
{\Sigma \vdash op : xs.\;S \lolli S'}
\]

\[
\infer
{\btsel{\alpha_1 + \alpha_2} : N_1 \with N_2}
{\alpha_1 : N_1 & \alpha_2 : N_2}
\qquad
\infer
{\btcond{S}{\alpha} : S \lolli S \tensor N}
{\alpha : N}
\qquad
\infer
{\btseq{\alpha_1; \alpha_2} : \mseq(N_1,N_2)}
{\alpha_1 : N_1 & \alpha_2 : N_2}
\qquad
\]

The $\mseq$ operator is defined as follows:
\begin{eqnarray*}
  \mseq(\one, N) &=& N\\
  \mseq(S_1, S_2) &=& S_1\tensor S_2\\
  \mseq(S, S'\tensor N) &=& (S\tensor S')\tensor N\\
  \mseq(S, N_1\with N_2) &=& \mseq(S, N_1)\with \mseq(S, N_2)\\
  \mseq(S_1, S_2 \lolli N) &=& S_1 \tensor (S_2 \lolli N)\\
  \mseq(S\tensor N_1, N_2)  &=& \mseq(S, \mseq(N_1, N_2))\\
  \mseq(S_1\lolli N_1, N_2) &=& S_1 \lolli \mseq(N_1, N_2)\\
  \mseq(N_1\with N_2, N)   &=& (\mseq(N_1,N)\with\mseq(N_2,N))
\end{eqnarray*}

It can be interpreted as pushing the requirements of the first formula to
the outside of the whole formula, then conjoining its consequences with the
specification of the second formula.  The correctness of this definition,
and of the type system in general, with respect to the operational
semantics, is considered next.

%

\subsection{Metatheory}

We revisit Conjecture~\ref{conj:typesound} and sketch a proof. First we
establish a lemma about the $\mseq$ operator:

\begin{lemma}
  \label{lemma:seq}
  If $\Delta, \mseq(N_1, N_2) \vdash S$ and $\Delta$ is \textit{flat}, i.e.
  consists only of propositions of the form $S$, then there exists $S_1$ such that
$\Delta, N_1 \vdash S_1$ and $\Delta, S_1 \vdash N_2$.
\end{lemma}

\begin{proof}
By induction on the definition of $\mseq$. We show the interesting cases.
\begin{itemize}
  \item \textbf{Case:} 
    $\mseq(S_1, S_2 \lolli N) = S_1 \tensor (S_2\lolli N)$.\\
    Assume $\Delta, S_1, S_2 \lolli N \vdash S$. In this case,
    we can just tensor together the first state and feed it into the
    second.
    $\Delta, S_1 \vdash \bigotimes(\Delta)\tensor S_1$, and
    $\Delta, \bigotimes(\Delta)\tensor S_1, S_2 \lolli N \vdash S$ by
    untensoring that proposition to get to the assumption.
  \item \textbf{Case:}
    $\mseq(S_1\lolli N_1, N_2) = S_1 \lolli \mseq(N_1, N_2)$.\\
    Assume $\Delta, S_1 \lolli \mseq(N_1, N_2) \vdash S$.
    Because the proof of this sequent concludes with an $S$, somewhere
    along the way we must discharge the $\lolli$, i.e. some part of
    $\Delta$ proves $S_1$. Rewrite $\Delta = \Delta_1, \Delta'$ where
    $\Delta_1 \vdash S_1$. Somewhere in the proof there is an application
    of $\lolli L$ such that $\Delta', \mseq(N_1, N_2) \vdash S$ is a
    subproof, and by inductive hypothesis, there exists $S'$ such that
    $\Delta', N_1 \vdash S'$ and $S', N_2 \vdash S$. 

    Now it suffices to show that $\Delta, S_1\lolli N_1 \vdash S'$ (since
    we already have $S', N_2 \vdash S$). This can be established by reusing
    the part of $\Delta$ that discharges $S_1$, using $\lolli L$ on 
    $\Delta_1 \vdash S_1$ and $\Delta', N_1 \vdash S'$.
  \item Remaining cases are straightforward.
\end{itemize}
\end{proof}

\begin{theorem}
If $\alpha : A$, $\Delta$ is flat, and $\Delta, A \vdash S$, then
$\eval{\alpha}{\Delta}{\Delta'}$ and $\Delta' \vdash S$.
\end{theorem}

\begin{proof}
By lexicographic induction on the typing derivation and proof. We show the
sequence case here.
\begin{itemize}
  \item \textbf{Case:}
    \[
      \infer{\btseq{\alpha_1; \alpha_2} : \mseq(N_1,N_2)}
      {\alpha_1 : N_1 & \alpha_2 : N_2}
    \]
    Known: $\Delta, \mseq(N_1,N_2) \vdash S$. By lemma, there exists $S'$
    such that $\Delta, N_1 \vdash S'$ and $S', N_2 \vdash S$.\\

    By i.h., $\eval{\alpha_1}{\Delta}{\Delta'}$ where $\Delta'\vdash S'$.
    By i.h., $\eval{\alpha_2}{\{S'\}}{\Delta''}$ where $\Delta'' \vdash S$.
    By appropriate equivalence between the positive propostion $S'$ and 
    context $\Delta'$, and by the sequence evaluation rule,
    $\eval{\btseq{\alpha_1;\alpha_2}}{\Delta}{\Delta''}$ where
    $\Delta'' \vdash S'$. 
\end{itemize}
\end{proof}

\subsection{Example}

We now return to the ``investigating a sound'' example whose evaluation was
shown in Section~\ref{sec:eval-example}. The computed type for the example:
\input{figures/intro-ex.tex}
is:
\begin{eqnarray*}
  & (\mathtt{heard\_noise} \lolli (\mathtt{heard\_noise}\lolli
  \mathtt{no\_target} \lolli \mathtt{has\_target})\\
\with & (\mathtt{has\_target} \lolli
  (\mathtt{has\_target}\tensor\mathtt{at\_target}\tensor \\
&
    (\mathtt{has\_target}\tensor\mathtt{at\_target}\tensor\mathtt{heard\_noise}
    \lolli \mathtt{no\_target}))\\
\with & (\mathtt{has\_cigarette}\lolli \one)\\
\with & (\one \lolli \one)
\end{eqnarray*}

%% file: impl.tex
We implemented both an interpreter for BTL (with repeater nodes) and a type synthesis algorithm for BTL excluding repeaters, both following the descriptions in the paper. The implementation is written in 523 lines of Standard ML, including detailed comments, and we authored an additional 448 lines of examples, including those used in this paper.

The implementation is freely available on GitHub at (URL redacted for double-blind review).

%% file: discussion.tex
A longer-term goal for this work is to be able to account for how behavior
trees are used in practice, to integrate the type system into the behavior
authoring process (perhaps through a combination of checking and synthesis), 
and to evaluate how it may make designers (particularly without programming
background) more effective. We anticipate using the implementation of BTs
in the Unreal Engine as a benchmark. Shorter term, there are a few theoretical
concerns we still need to consider. We now describe a roadmap for this
project.

\subsection{Parallel Composition}

\begin{wrapfigure}[17]{R}{0.5\textwidth}
\centering
\includegraphics[width=0.45\textwidth]{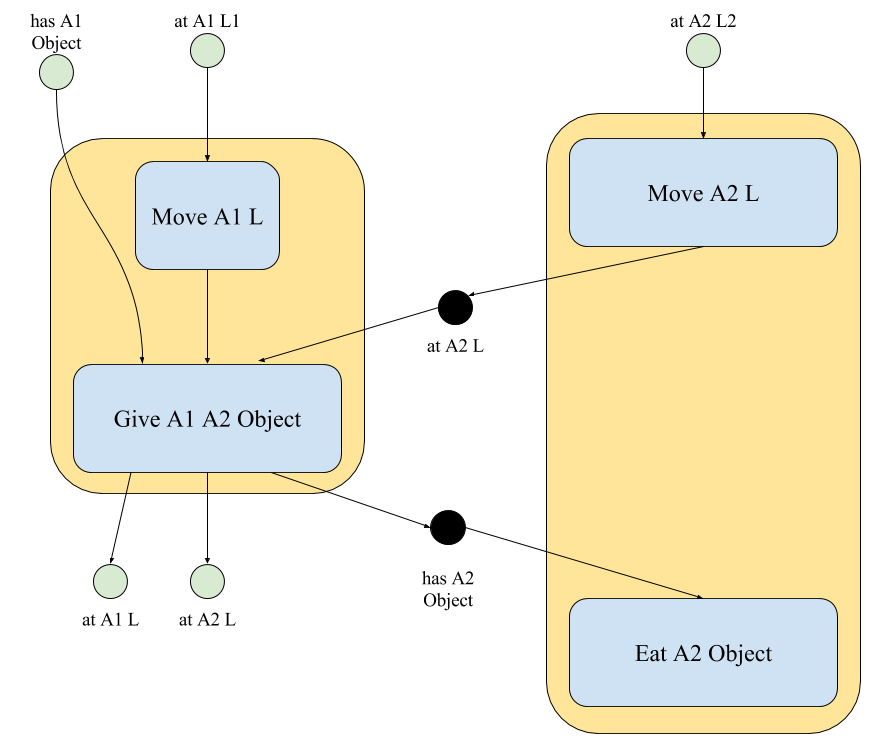}
\caption{Composing processes that interact.}
\label{fig:interaction}
\end{wrapfigure}
\leavevmode

\newcommand{\btpar}[1]{\mathsf{Par}\{#1\}}

Currently, the semantics of agents operating in the world concurrently is
not specified by the language.  To account for placing multiple
world-manipulating agents into the environment, we might consider
introducing a ``parallel'' operator to BTL:

$\alpha ::= \ldots  \mid \btpar{\alpha_1 \parallel \alpha_2}$

We may consider a few options for an operational semantics that warrant a
different type-theoretic treatment. For instance, perhaps parallel
behaviors split the state and operate in isolation until complete. This
behavior could be captured with the rule:
\begin{flalign*} 
&  \infer
{ \eval{\btpar{\alpha_1\parallel \alpha_2}}{\Delta}{\Delta_1', \Delta_2'}}
{ \Delta = \Delta_1, \Delta_2 &
\eval{\alpha_1}{\Delta_1}{\Delta_1'} &
\eval{\alpha_2}{\Delta_2}{\Delta_2'}
} &
\end{flalign*}

Additional rules may specify that if either subbehavior fails, the whole
behavior fails.

However, in practice, behavior trees allow for finer-grained interactions
between processes. The above specification precludes, for example, the
below two behaviors succeeding:

\begin{lstlisting}[frame = single]
// Agent 1 (a1)       // Agent 2 (a2)
   Seq{move(a1,L);       Seq{move(a2,L);
       give(a2,O)}           eat(a2,O)}
\end{lstlisting}

These behaviors will only succeed if they interact when run;
\lstinline|a1|'s action \lstinline|give(a2,O)| will only succeed if
\lstinline|a2|'s first action, \lstinline|move(a2,L)|, is permitted to
succeed first.
Figure~\ref{fig:interaction} describes visually the behavior specification
we would like to result in this interaction.

\newcommand{\process}[2]{#1/#2}
\newcommand{\steps}{\rightarrow}
\newcommand{\stepsalong}[1]{\xrightarrow{#1}}
\newcommand{\notsteps}{\not\rightarrow}

\newcommand{\cstep}{\rightsquigarrow}
\newcommand{\lpre}{\preceq}
\newcommand{\proves}{\Rightarrow}

To account for such fine-grained concurrent behaviors formally, we require a small-step
semantics over the judgment $\process{\alpha}{\Delta} \steps
\process{\alpha}{\Delta'}$. A sketch of this semantics that includes
parallel composition is in
Figure~\ref{fig:smallstep}. However, note that this semantics does not
properly handle failure; instead, it embodies the synchronous semantics of
behaviors simply pausing (failing to evolve) if their conditions are not
satisfied, instead permitting the possibility of a delayed transition if
their conditions become satisfied as another behavior evolves. While this
behavior may be useful in some scenarios, it is not universally desirable,
so we need a way to account for this behavior, perhaps through a
stack-based semantics with success and failure continuations. Likewise, the
type system has a clear extension to count for arbitrarily ``pausing''
processes ($\tensor$ is a straightforward interpretation), but accounting
for failure in the type system is also left to work.

\begin{figure} 
\fbox{%
\parbox{0.98\linewidth}{%

\[
\infer[step/\mathsf{skip}]
{\process{\Delta}{\mathsf{skip}} \steps \process{\Delta}}{}
{}
\qquad
\infer[step/op]
	{ \process{\Delta, \Delta_{A}}{op(args)} \steps \process{\Delta, B}}
    { \Sigma \vdash op[args] : A \lolli B 
    &
      \Delta_{A} \vdash A
    }
\]

\[
\infer[step/;]
{\process{\Delta}{\alpha_1;\alpha_2} \steps
			\process{\Delta'}{\alpha_2}
}
{
\process{\Delta}{\alpha_1} \steps^* \process{\Delta'}{}
}
{}
\qquad
\infer[step/+_1]
{\process{\Delta}{\alpha_1 + \alpha_2} \steps
	\process{\Delta}{\alpha_1}}{}
{}
\qquad
\infer[step/+_2]
{\process{\Delta}{\alpha_1 + \alpha_2} \steps \process{\Delta}{\alpha_2}}
{\process{\Delta}{\alpha_1} \notsteps}
\]

\[
\infer[step/?]
	{\process{\Delta, p}{?p.\;\alpha} \steps \process{\Delta, p}{\alpha}}
	{
     }
{}
\qquad
{}
\qquad
\infer[step/^*]
{\process{\Delta}{\alpha^*} \steps \process{\Delta'}{\alpha^*}}
{\process{\Delta}{\alpha} \steps^* \process{\Delta'}{}}
\]

\[
\infer[step/{\parallel}_1]
	{\process{\Delta_1, \Delta_2}{\actpar{\alpha_1}{\alpha_2}} 
    	\steps
     \process{\Delta_1',\Delta_2}{\actpar{\alpha_1'}{\alpha_2}}}
    {\process{\Delta_1}{\alpha_1}
    	\steps \process{\Delta_1'}{\alpha_1'}}
{}\qquad
\infer[step/{\parallel}_2]
{\process{\Delta_1, \Delta_2}{\actpar{\alpha_1}{\alpha_2}}
	\steps
  \process{\Delta_1, \Delta_2}{\actpar{\alpha_1}{\alpha_2'}}}
{\process{\Delta_2}{\alpha_2} \steps \process{\Delta_2'}{\alpha_2'}
}
\]

}}
\caption{A small-step semantics for BTL without failure.}
\label{fig:smallstep}
\end{figure}

\subsection{Theoretical extensions}

In addition to accounting for parallel execution, we also need to consider
repeater nodes. The operational semantics are fairly easy to specify, but
guaranteeing convergence of computing fixed points for a type-based
characterization may prove difficult. Recursive types have been
successfully integrated into linear logic~\cite{baelde2012least}, and we
plan to investigate their use, although readability may remain a challenge.

Another step we would like to take is to introduce additional forms of lightweight verification
on top of the type system. For instance, selectors are often designed with
the intent that all possible cases are covered: each child of the selector
is guarded by a condition, and the disjunction of the conditions
characterizes an {\em invariant} of the state. Provided a proof that the
invariant actually holds, it may be useful to simplify the type to omit the
guards. This corresponds to provability between e.g. 
$(A\oplus B) \tensor ((A \lolli C) \with (B \lolli D))$ and
$(C \oplus D)$.

Next, while we have established a correspondence between the type system
and evaluation of {\em successful} behaviors, we believe we can formulate a
conjecture to the effect that the situations in which types fail to yield a
flat context (because there is some implication that cannot be discharged
on the left, say) correspond to the failure cases of execution. We expect
this proof will be more difficult than the former.

%
%
%
%
%

%% file: conclusion.tex
We have presented a formal semantics and type system for a fragment of behavior trees as
they are used to describe characters in virtual environments. Our system
includes a reference implementation and correctness proofs.  This work
represents substantial new ground broken towards a longer-term vision of
authoring robust, designer-friendly specifications for reactive agent
behaviors.

If our long-term vision is successful, we can enable several new things
for behavior authors: integrating hand-authored trees with behavior
synthesis algorithms through linear logic theorem proving (akin to
planning); the development of {\em behavior libraries} consisting of
reusable, parameterized behavior trees whose specification precludes
examining the entire tree; and certified behaviors that are guaranteed to
succeed under certain environments. These features would improve the
effectiveness of developing agents in virtual worlds with varied
applications in entertainment, arts, simulation modeling, research
competitions and challenges, and computing education.